\documentclass[11pt]{article}

\usepackage{amsmath}
\usepackage{amsthm}
\usepackage{amssymb}
\usepackage{amsfonts}
\usepackage{graphicx}
\usepackage{epsfig}
\usepackage{epsf}
\usepackage{epstopdf}
\usepackage{datetime}
\usepackage{latexsym}
\usepackage{enumerate}
\usepackage[small]{caption}
\usepackage{amsfonts}
\usepackage{fancyvrb}
\usepackage{fancyhdr}
\usepackage{amssymb}
\usepackage{amsthm}
\usepackage{graphics}
\usepackage{psfrag}
\usepackage{bbm}
\usepackage{url}
\usepackage{hyperref}

\newtheorem{theorem}{Theorem}
\newtheorem{lemma}{Lemma}

\newtheorem{corollary}{Corollary}

\newtheorem{observation}{Observation}

\newtheorem{problem}{Problem}
\newtheorem{question}{Question}

\def\etal{{et~al.}}
\def\eg{{e.g.}}
\def\ie{{i.e.}}

\newcommand{\alg}{\textsf{ALG}}
\newcommand{\opt}{\textsf{OPT}}

\newcommand{\dist}{{\rm dist}}

\newcommand{\NN}{\mathbb{N}} 
\newcommand{\ZZ}{\mathbb{Z}} 
\newcommand{\RR}{\mathbb{R}} 
\newcommand{\EE}{\mathbb{E}} 
\newcommand{\eps}{\varepsilon}

\def\B{\mathcal B}
\def\C{\mathcal C}

\def\Q{\mathcal Q}

\newcommand{\old}[1]{{}}
\newcommand{\later}[1]{{}}

\title{\textsc{Online Unit Clustering and Unit Covering \\
    in Higher Dimensions}
  \thanks{A preliminary version of this paper appeared in the
    \emph{Proceedings of the 15th Workshop on Approximation and Online Algorithms (WAOA)},
    LNCS~10787, Springer, Cham, 2017, pp.~238--252.}}

\author{Adrian Dumitrescu\thanks{
Algoresearch L.L.C., Milwaukee, WI, USA\@. Email:~\texttt{ad.dumitrescu@gmail.com}.}
\qquad
Csaba D. T\'oth\thanks{Department of Mathematics,
    California State University Northridge, Los Angeles, CA;
and Department of Computer Science, Tufts University, Medford, MA, USA\@.
Email:~\texttt{cdtoth@acm.org}.}
}

\begin{document}

\maketitle

\begin{abstract}
We revisit the online \textsc{Unit Clustering} and \textsc{Unit Covering}
problems in higher dimensions: Given a set of $n$ points in a metric space,
that arrive one by one, \textsc{Unit Clustering} asks to \emph{partition}
the points into the minimum number of clusters (subsets) of diameter at most one;
whereas \textsc{Unit Covering} asks to \emph{cover} all points by the minimum number
of balls of unit radius.
In this paper, we work in $\RR^d$ using the $L_\infty$ norm.

We show that the competitive ratio of any online algorithm (deterministic or randomized)
for \textsc{Unit Clustering} is $\Omega(d)$. In particular, it depends on the dimension $d$,
and this resolves an open problem raised by Epstein and van Stee (2008).
We also give a randomized online algorithm with competitive ratio $O(d^2)$
for \textsc{Unit Clustering} of integer points (\ie, points in $\ZZ^d$,
$d\in \NN$, under the $L_{\infty}$ norm).

We show that the competitive ratio of any deterministic online algorithm
for \textsc{Unit Covering} is at least $2^d$. This ratio is the best possible,
as it can be attained by a simple deterministic algorithm that assigns points
to a predefined set of unit hypercubes.
We complement these results with some additional lower bounds for
related problems in higher dimensions.


\end{abstract}

\section{Introduction} \label{sec:intro}

Covering and clustering are ubiquitous problems in the theory of algorithms,
computational geometry, optimization, and others.
Such problems can be asked in any metric space, however this generality
often restricts the quality of the results, particularly for online algorithms.
Here we study lower bounds for several such problems
in $\RR^d$, for a positive integer $d$, under the $L_\infty$ norm.
Recall that a ball under the $L_\infty$ norm is an axis-aligned hypercube.
We first consider their \emph{offline} versions.

\begin{problem} \label{prob:1}
\textsc{$k$-Center}. Given a set of $n$ points in $\RR^d$ and an integer $k$,
cover the set by $k$ congruent balls centered at the points so that the diameter
of the balls is minimized.
\end{problem}

The following two problems are dual to Problem~\ref{prob:1}.

\begin{problem}\label{prob:2}
\textsc{Unit Covering}. Given a set of $n$ points in $\RR^d$,
cover the set by balls of unit diameter so that the number of balls is minimized.
\end{problem}

\begin{problem}\label{prob:3}
\textsc{Unit Clustering}.
Given a set of $n$ points in $\RR^d$,
partition the set into clusters of diameter at most one so that the number of clusters is minimized.
\end{problem}

Problems~\ref{prob:1} and~\ref{prob:2} are easily solved in polynomial time for points on the line
($d=1$); however, both problems become NP-hard already in the Euclidean plane~\cite{FPT81,MS84}.
Factor $2$ approximations are known for \textsc{$k$-Center} in any metric space (and so for any
dimension)~\cite{FG88,Go85}; see also~\cite[Ch.~5]{Va01},~\cite[Ch.~2]{WS11},
while polynomial-time approximation schemes are known for \textsc{Unit Covering}
for any fixed dimension~\cite{HM85}. However, these algorithms are notoriously inefficient
and thereby impractical; see also~\cite{BLMS17} for a summary of results and different
time vs. ratio trade-offs.

Problems~\ref{prob:2} and~\ref{prob:3} look similar; indeed, one can go from balls
to clusters and vice versa in a straightforward way:
The balls in a unit covering form unit clusters if we assign multiply covered points
to unique balls. Conversely, the points in a unit cluster are contained in a unit ball
under $L_\infty$ norm, as the $x_i$-coordinates of the points differ by at most 1 for $i=1,\ldots , d$.
As such, the two problems are identical in the offline setting under the
$L_\infty$ norm\footnote{Problems~\ref{prob:2} and~\ref{prob:3} are equivalent
  under any norm in which every set of unit diameter is contained in a ball
  of unit diameter. This holds under the $L_1$ and $L_\infty$ norms,
  but not under the $L_p$ norm for any $1<p<\infty$ in $\RR^d$, $d\geq 2$.}.

We next consider their \emph{online} versions. In this paper we focus
on Problems~\ref{prob:2} and~\ref{prob:3} in particular.
It is worth emphasizing two common properties:
(i)~a point assigned to a cluster must remain in that cluster; and
(ii)~two distinct clusters cannot merge into one cluster, \ie, the clusters maintain their identities.

The performance of an online algorithm $\alg$ is measured by comparing it to an
optimal offline algorithm $\opt$ using the standard notion of competitive ratio~\cite[Ch.~1]{BY98}.
The competitive ratio of $\alg$ is defined as
$\sup_\sigma \frac{\alg(\sigma)}{\opt(\sigma)}$,
where $\sigma$ is an input sequence of request points,
$\opt(\sigma)$ is the cost of an optimal offline algorithm for $\sigma$
and $\alg(\sigma)$ denotes the cost of the solution produced by $\alg$ for this input.
For randomized algorithms, $\alg(\sigma)$ is  replaced by the expectation $E[\alg(\sigma)]$,
and the competitive ratio of $\alg$ is $\sup_\sigma \frac{E[\alg(\sigma)]}{\opt(\sigma)}$.
Whenever there is no danger of confusion, we use $\alg$ to refer to an algorithm
or the cost of its solution, as needed.

When discussing lower bounds for a \emph{randomized} online algorithm,
one can distinguish between two types of adversaries~\cite{BBKTW94}.
An \emph{adaptive online adversary} constructs the next input item (\eg, point) online,
based on the previous input items and previous actions of the algorithm.
In contrast, an \emph{oblivious adversary} must construct the entire input sequence in advance,
without having access to the actions of the algorithm. Obviously, an adaptive adversary is
more powerful, hence the competitive ratio against an adaptive adversary is greater or equal than
against an oblivious adversary.
Unless specified otherwise, upper bounds on the competitive ratio for randomized online algorithms
assume an adaptive adversary, and lower bounds an oblivious adversary.
Note, however, that for \emph{deterministic} online algorithms,
the competitive ratio is the same under both adversarial models.

\paragraph{Related previous work.}
Charikar~\etal~\cite{CCFM04} studied the online version of \textsc{Unit Covering}.
The points arrive one by one and each point needs to be assigned to a new or to an
existing unit ball upon arrival; the $L_2$ norm is used in $\RR^d$, $d\in \NN$.
The location of each new ball is fixed as soon as it is opened.
The authors provided a deterministic algorithm of competitive ratio $O(2^d d \log{d})$
and gave a lower bound of $\Omega(\log{d} / \log{\log{\log{d}}})$
on the competitive ratio of any deterministic online algorithm for this problem.

Recently, Dumitrescu, Ghosh, and T\'oth~\cite{DGT18} showed that the competitive ratio
of \texttt{Algorithm Centered} for online \textsc{Unit Covering} in $\RR^d$, $d\in \NN$,
under the $L_{2}$ norm is bounded by the Newton number\footnote{For a convex body
  $C \subset \RR^d$, the \emph{Newton number}
  (a.k.a.~\emph{kissing number}) of $C$ is the maximum number of nonoverlapping
  congruent copies of $C$ that can be arranged around $C$ so that they each touch
  $C$~\cite[Sec.~2.4]{BMP05}.} of the Euclidean ball in the same dimension.
In particular, this ratio is $O(1.321^d)$.
They also established a lower bound of $d+1$ for every $d \geq 1$ (and $4$ for $d=2$).

Chan and Zarrabi-Zadeh~\cite{CZ09} introduced the online  \textsc{Unit Clustering} problem.
Whereas the input and the objective of this problem are identical to those
for \textsc{Unit Covering}, this latter problem is more flexible in that
the algorithm is not required to produce unit balls at any time, but rather
the smallest enclosing ball of each cluster should have diameter \emph{at most} $1$;
moreover, a ball may change (grow or shift) in time.
The $L_\infty$ norm is used in $\RR^d$, $d\in \NN$.
The authors showed that several standard approaches for \textsc{Unit Clustering}, namely
\texttt{Algorithm Centered},
\texttt{Algorithm Grid}, and
\texttt{Algorithm Greedy},
all have competitive ratio at most $2$ for points on the line ($d=1$).
Moreover, the first two algorithms above are applicable for \textsc{Unit Covering},
with a competitive ratio at most $2$ for $d=1$, as well.

\smallskip
In fact, Chan and Zarrabi-Zadeh~\cite{CZ09} showed that no online algorithm
(deterministic or randomized) for \textsc{Unit Covering}
can have a competitive ratio better than $2$ in one dimension ($d=1$).
They also showed that it is possible to get better results
for \textsc{Unit Clustering} than for \textsc{Unit Covering}.
Specifically, they devised  the first algorithm with competitive ratio below $2$ for $d=1$,
namely a randomized algorithm with competitive ratio $15/8$; they further
improved this ratio to $11/6$~\cite{ZC09}.
Moreover, they developed a general method to achieve competitive ratio below $2^d$
in $\RR^d$ under the $L_{\infty}$ norm for any $d \geq 2$, by lifting
the one-dimensional algorithm to higher dimensions. In particular,
the existence of an algorithm for \textsc{Unit Clustering} with competitive ratio
$\rho_1$ for $d=1$ yields an algorithm with competitive ratio $\rho_d = 2^{d-1} \rho_1$
for every $d \geq 2$ for this problem.  The current best competitive ratio
for \textsc{Unit Clustering} in $\RR^d$, $2^{d-1} \frac53$ for every $d \geq 2$,
is obtained in exactly this way (by lifting the algorithm of Ehmsen and Larsen~\cite{EL13}).

A simple deterministic algorithm (\texttt{Algorithm Grid} below) that assigns points
to a predefined set of unit cubes that partition $\RR^d$ can be easily proven to be
$2^d$-competitive for both \textsc{Unit Covering} and \textsc{Unit Clustering}.
Since each cluster of $\opt$ can be split into at most $2^d$ grid-cell clusters created
by the algorithm, the competitive ratio of \texttt{Algorithm Grid} is at most $2^d$,
and this analysis is tight. See Fig.~\ref{fig:f3} for an example in the plane.

\begin{quote}
\texttt{Algorithm Grid.} Build a uniform grid in $\RR^d$ where cells are
unit cubes of the form \linebreak
$\prod_{j=1}^d \, [i_j,i_j+1)$, where $i_j \in \ZZ$ for $j=1,\ldots,d$.
For each new point $p$, if the grid cell containing $p$ is nonempty,
put $p$ in the corresponding cluster; otherwise open a new cluster for the grid cell
and put $p$ in it.
\end{quote}

\begin{figure}[htbp]
\begin{center}
  \includegraphics[scale=1.2]{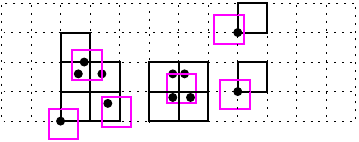}
\end{center}
\caption{Example for \texttt{Algorithm Grid} in the plane; here $\alg=11$ and $\opt=6$.}
\label{fig:f3}
\end{figure}

We summarize the current best online algorithms for \textsc{Unit Clustering}
in low dimensions; see Table~\ref{table:CompetitiveRatios}.
For $d=1$, the current best ratio, $5/3$, is due to Ehmsen and Larsen~\cite{EL13}
and is produced by a deterministic algorithm;
on the other hand, the current best lower bound for deterministic algorithms, $13/8$,
is due to Kawahara and Kobayashi~\cite{KK15}.
The current best lower bound for randomized algorithms, $3/2$, is due to
Epstein and van Stee~\cite{ES10}.

For $d=2$, the current best ratio, $10/3$, follows from lifting the algorithm of
Ehmsen and Larsen~\cite{EL13} from $d=1$ to $d=2$ by using the technique of
Chan and Zarrabi-Zadeh~\cite{CZ09} mentioned earlier.
The current best lower bound for deterministic algorithms, $13/6$,
is due to Ehmsen and Larsen~\cite{EL13}.
The current best lower bound for randomized algorithms, $11/6$, is due to
Epstein and van Stee~\cite{ES10}.

\begin{table}[htbp]
\begin{center}
\bgroup
\def\arraystretch{1.15}
\begin{tabular}{|l|l|l|l|l|}
\hline
 & \textsc{Unit Covering} & \textsc{Unit Covering} &\textsc{Unit Clustering} &  \textsc{Unit Clustering}\\
Dimension & lower bound & upper bound  & lower bound & upper bound \\
\hline \hline
$d=1$  & $2$~\cite{CCFM04} & $2$~\cite{CZ09}  & $13/8$~\cite{KK15} & $5/3$~\cite{EL13}\\
\hline
$d=2$  & 4 [$\star$] & $4$~\cite{CZ09}  & $13/6$~\cite{EL13} & $10/3$~\cite{CZ09,EL13}\\
\hline
$d\geq 3$ & $2^d$ [$\star$] & $2^d$ \cite{CZ09} & $\Omega(d)$ [$\star$] & $\frac{5}{3} \cdot 2^{d-1}$~\cite{CZ09,EL13}\\
\hline
\end{tabular}
\egroup
\end{center}
\caption{Current best bounds on the competitive ratio of deterministic algorithms for
  online \textsc{Unit Covering} and \textsc{Unit Clustering}
  in $\RR^d$ under $L_\infty$ norm for $d=1$, $d=2$, and $d\geq 3$.
  New results are labeled with [$\star$].}\label{table:CompetitiveRatios}
\end{table}

\paragraph{Notation and terminology.}
Throughout this paper the $L_\infty$ norm is used in $\RR^d$ ($d \geq 1$).
A~\emph{hyperrectangle} in $\RR^d$ is the Cartesian product of $d$ closed intervals $R=\prod_{i=1}^d[a_i,b_i]$,
where the lengths $b_i-a_i$ of the intervals, for $i=1,\ldots ,d$, are the \emph{extents} of $R$.
A hyperrectangle is a \emph{hypercube} (or \emph{cube}, for short) if all $d$ extents have the same length,
and a \emph{unit cube} if all $d$ extents have unit length.
For a vector $\mathbf{x}\in \RR^d$
and a set $S\subset \RR^d$, we denote by $\textbf{x}+S=\{\mathbf{x}+\mathbf{s}: \mathbf{s}\in S\}$
the translate of $S$ by vector $\mathbf{x}$. In particular, every unit cube in $\RR^d$ can be written
in the form $\mathbf{x}+[0,1]^d$ for some $\mathbf{x} \in \RR^d$.
For a random variable $X$, $\EE[X]$ denotes its expected value.

\paragraph{Contributions.} We obtain the following results:
\begin{enumerate}[(i)]

\item The competitive ratio of every online algorithm (deterministic or randomized)
for \textsc{Unit Clustering} in $\RR^d$ under the $L_{\infty}$ norm is $\Omega(d)$ for
every $d\geq 2$ (Theorem~\ref{thm:unit-clustering} in Section~\ref{sec:unit-clustering}).
We thereby give a positive answer to a question of Epstein and van Stee; specifically,
they asked whether the competitive ratio grows with the dimension~\cite[Sec.~4]{ES10}.
The question was reposed in~\cite[Sec.~7]{EL13}.

\item The competitive ratio of every deterministic online algorithm
for \textsc{Unit Covering} in $\RR^d$ under the $L_{\infty}$ norm is at least $2^d$ for every $d \geq 1$.
This bound cannot be improved; as such, \texttt{Algorithm Grid} is optimal
in this setting (Theorem~\ref{thm:2^d} in Section~\ref{sec:unit-covering-r^d}).
This generalizes a result by Chan and Zarrabi-Zadeh~\cite{CZ09} from $d=1$
to higher dimensions.

\item The competitive ratio of every deterministic online algorithm
for \textsc{Unit Covering} in $\ZZ^d$ under the $L_{\infty}$ norm is at least $d+1$ for every $d \geq 1$
(Theorem~\ref{thm:d+1} in Section~\ref{sec:unit-covering-z^d}).

\item We give a randomized algorithm with competitive ratio $O(d^2)$
for \textsc{Unit Covering} in $\ZZ^d$, $d\in \NN$, under the $L_{\infty}$ norm
(Theorem~\ref{thm:covering-integer} in Section~\ref{sec:unit-covering-z^d}).
The algorithm applies to \textsc{Unit Clustering} in $\ZZ^d$, $d\in \NN$,
with the same competitive ratio.

\item The competitive ratio of \texttt{Algorithm Greedy}
for \textsc{Unit Clustering} in $\RR^d$ under the $L_{\infty}$ norm is unbounded for every $d\geq 2$
(Theorem~\ref{thm:greedy} in Section~\ref{sec:greedy}).
The competitive ratio of \texttt{Algorithm Greedy}
for \textsc{Unit Clustering} in $\ZZ^d$ under the $L_{\infty}$ norm is at least $2^{d-1}$
and at most $2^{d-1}+\frac12$ for every $d \geq 2$
(Theorem~\ref{thm:greedy-integer} in Section~\ref{sec:greedy}).

\end{enumerate}

\paragraph{Broader Perspective.}
Several other variants of \textsc{Unit Clustering} have been studied in~\cite{ELS08}.
A survey of algorithms for \textsc{Unit Clustering} in the context of online algorithms
appears in~\cite{Ch08}; see also~\cite{Du18} for a review overview.
Clustering with variable sized clusters has been studied
in~\cite{CEIL13,DI13}. Grid-based online algorithms for clustering problems have
been developed by the same authors~\cite{DI14}.

\textsc{Unit Covering} is a variant of \textsc{Set Cover}. Alon~\etal~\cite{AAA+09}
gave a deterministic online algorithm of competitive ratio $O(\log{m} \log{n})$
for \textsc{Set Cover}, where $n$ is the number of possible points
(the size of the ground set) and $m$ is the number of sets in the family.
If every element appears in at most $\Delta$ sets, the competitive ratio of the algorithm
can be improved to $O(\log{\Delta} \log{n})$.
Buchbinder and Naor~\cite{BN09b} improved these competitive ratio to $O(\log m \log(n/\opt))$
and $O(\log{\Delta} \log{(n/\opt)})$, respectively, under the same assumptions.
For several combinatorial optimization problems (e.g., covering and packing), the classic
technique that rounds a fractional linear programming solution to an integer solution
has been adapted to the online setting~\cite{ABC+16,ABFP13,ACR17,BN09b,GN14}.

In these results, the underlying set system for the covering and packing problem must be finite:
The online algorithms and their analyses rely on the size of the ground set. For
\textsc{Unit Clustering} and \textsc{Unit Clustering} over infinite sets,
such as $\RR^d$ or $\ZZ^d$, these techniques could only be used after a suitable
discretization and a covering of the domain with finite sets, and it is unclear
whether they can beat the trivial competitive ratio of $2^d$ in a substantive way.

\section{Lower bound for online  \textsc{Unit Clustering}} \label{sec:unit-clustering}

In this section, we prove the following theorem.

\begin{theorem} \label{thm:unit-clustering}
The competitive ratio of every
(i)~deterministic algorithm, and (ii)~randomized algorithm,
for \textsc{Unit Clustering} in $\RR^d$ under the $L_{\infty}$ norm is $\Omega(d)$
for every $d\geq 1$.
\end{theorem}

\begin{proof}
Let $\varrho$ be the competitive ratio of an online algorithm.
We may assume that $\varrho \leq d$, otherwise there is nothing to prove.
We may also assume that $d\geq 4$ since this is the smallest value for which
the argument gives a nontrivial lower bound.
Let $K$ be a sufficiently large even integer (that depends on $d$).

\paragraph{Deterministic Algorithm.}
We first prove a lower bound for a deterministic algorithm,
assuming without loss of generality an adaptive deterministic adversary.
We present a total of $\lfloor d/2\rfloor K^d$ points to the algorithm,
and show that it creates $\Omega(d \cdot \opt)$ clusters,
where $\opt$ is the offline minimum number of clusters for the final set of points.
Specifically, we present the points to the algorithm in $\lfloor d/2\rfloor$ rounds.
Round $i=1,\ldots , \lfloor d/2\rfloor$ consists of the following three events:
\begin{enumerate} [(i)] \itemsep -2pt
\item The adversary presents (inserts) a set $S_i$ of $K^d$ points;
  $S_i$ is determined by a vector $\sigma(i)\in \{-1,0,1\}^d$ to be defined later.
  We denote by $S_{\leq i}=\bigcup_{j=1}^i S_i$ the set of points presented so far.
\item The algorithm creates new clusters or expands existing clusters to cover $S_i$.
\item If $i< \lfloor d/2\rfloor$,
 the adversary computes $\sigma(i+1)$ from the clusters that cover $S_i$.
\end{enumerate}

In the first round, the adversary presents points of the integer lattice; namely
$S_1=[K]^d$, where $[K]=\{x\in \ZZ: 1\leq x\leq K\}$.
In round $i=2,\ldots ,\lfloor d/2\rfloor$, the point set $S_i$ will depend on the clusters
created by the algorithm in previous rounds.
We say that a cluster \emph{expires in round $i$}
if it contains some points from $S_i$ but no additional points can (or will) be added to it in any
subsequent round. We show that over $\lfloor d/2\rfloor$ rounds,
$\Omega(d\cdot \opt)$ clusters expire, which readily implies $\varrho= \Omega(d)$.

\paragraph{Optimal solutions.}
For $i=1,\ldots ,\lfloor d/2\rfloor$, denote by $\opt_i$ the offline optimum
for the set $S_{\leq i}$ of points presented up to round $i$.
Since $S_1=[K]^d$ and $K$ is even, $\opt_1=K^d/2^d$.
The optimum solution for $S_1$ is unique, and each cluster in the optimum solution
is a Cartesian product $\prod_{i=1}^d \{a_i,a_i+1\}$, where $a_i\in [K]$ is odd
for $i=1,\ldots,d$ (Fig.~\ref{fig:perturbation}(a)).

\paragraph{Near-optimal solutions.}
Consider $2^d-1$ additional near-optimal solutions for $S_1$ obtained by translating
the optimal clusters by a $d$-dimensional $0-1$ vector, and adding new clusters
along the boundary of the cube $[K]^d$.  We shall argue that the points inserted
in round $i$, $i\geq 2$, can be added to some but not all of these clusters.
To make this precise, we formally define these solutions for the integer grid $S_1$.
First we define an infinite set of hypercubes
$$\Q= \left\{\prod_{i=1}^d [a_i,a_i+1]: a_i\in \ZZ
\mbox{ \rm is odd for }i=1,\ldots ,d\right\}. $$
For any point set $S\subset \RR^d$ and a vector $\tau\in \{0,1\}^d$, let the clusters $C(S,\tau)$ be
the subsets of $S$ that lie in translates $Q+\tau$ of hypercubes $Q\in \Q$,
that is, let
$$ C(S,\tau)=\{ S\cap (Q+\tau): Q\in \Q\}.  $$

In general, the clusters $C(S,\tau)$ do not contain all points in $S$. However,
since $S_1$ is an integer grid, the clusters $C(S_1,\tau)$ contain all points in $S_1$
for every $\tau\in \{0,1\}^d$. See Fig.~\ref{fig:perturbation}(a--d) for examples.
Note that if $S\cap (Q+\tau)\neq\emptyset$ for some $Q\in \mathcal{Q}$, then
every point in $Q+\tau$ is within unit distance from $S$, and in particular,
$\ZZ^d\cap (Q+\tau)$ comprises $2^d$ points with coordinates in $\{0,1,\ldots , K+1\}$.
Consequently, the number of clusters in $C(S_1,\tau)$ is at most
$$\frac{(K+2)^d}{2^d} = \frac{K^d+O(dK^{d-1})}{2^d}
=\opt_1 \cdot \left(1+O\left(\frac{d}{K}\right)\right) = (1+o(1)) \, \opt_1,$$
if $K$ is sufficiently large with respect to $d$.

In round $i=2,\ldots ,\lfloor d/2\rfloor$, the point set $S_i$ is a perturbation of
the integer grid $S_1$ (as described below). Further, we will ensure
(cf.,~Observation~\ref{obs:1} below) that for every $i=1,\ldots , \lfloor d/2\rfloor$,
the point set $S_{\leq i}$ is covered by the clusters $C(S_1,\tau)$ for some $\tau\in \{0,1\}^d$.
In particular,  the final point set $S_{\leq \lfloor d/2\rfloor}$
is covered by the clusters $C(S_1,\tau)$ for some $\tau\in \{0,1\}^d$.
Consequently,
\begin{equation}   \label{eq:opt_i}
  \opt_i=\opt_1(1+o(1)) = (1+o(1)) \frac{K^d}{2^d}, \text{ for all } i=1,\ldots , \lfloor d/2\rfloor.
\end{equation}
At the end, we have $\opt=\opt_{\lfloor d/2\rfloor}= (1+o(1)) \frac{K^d}{2^d}$.

\paragraph{Perturbation.}
A perturbation of the integer grid $S_1$ is encoded by a vector $\sigma\in \{-1,0,1\}^d$,
that we call the \emph{signature} of the perturbation. Let $\eps\in (0,\frac{1}{2})$.
For an integer point $p=(p_1,\ldots , p_d)\in S_1$ and a signature $\sigma$,
the perturbed point $p'=(p_1',\ldots , p_d')$ is defined as follows;
see Fig.~\ref{fig:perturbation}(e--h) for examples in the plane:
For $j=1,\ldots , d,$ let
\begin{itemize}\itemsep -2pt
\item $p_j'=p_j$ when $\sigma_j=0$;
\item $p_j'=p_j+\eps$ if $p_j$ is odd, and $p_j'=p_j-\eps$ if $p_j$ is even when $\sigma_j=-1$;
\item $p_j'=p_j+\eps$ if $p_j$ is even, and $p_j'=p_j-\eps$ if $p_j$ is odd when $\sigma_j=1$.
\end{itemize}

For $i=2,\ldots ,\lfloor d/2\rfloor$, the point set $S_i$ is a perturbation of $S_1$
with signature $\sigma(i)\in \{-1,0,1\}^d$.
The signature of $S_1$ is $\sigma(1)=(0,\ldots ,0)$ (and so $S_1$ can be viewed as a
null perturbation of itself).
At the end of round $i=1,\ldots, \lfloor d/2\rfloor -1$, we compute $\sigma(i+1)$ from $\sigma(i)$
and the clusters that cover $S_i$ (as described below).
The signature $\sigma(i)$ determines the set $S_i$, for every $i=2,\ldots ,\lfloor d/2\rfloor$.
Note the following relation between the signatures $\sigma(i)$ and the clusters $C(S_i,\tau)$.

\begin{figure}[htpb]
\centering
\includegraphics[width=0.99\textwidth]{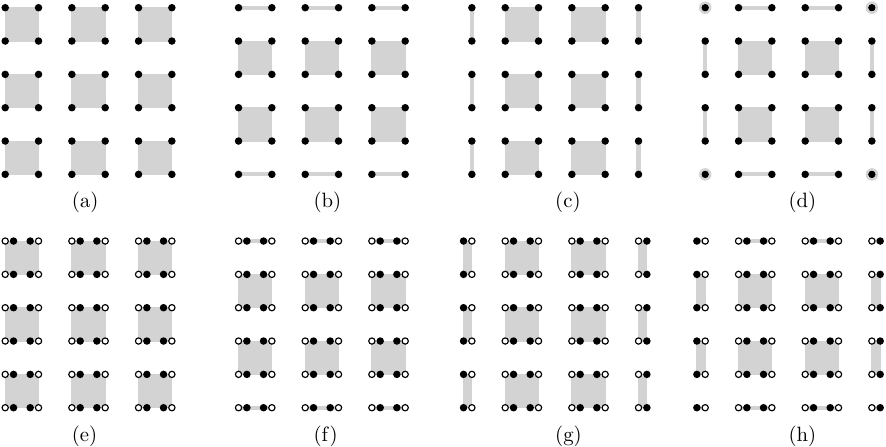}
\caption{(a) A $6\times 6$ section of the integer grid and $\opt_1=9$ clusters.
  (b--d) Near-optimal solutions $C(S_1,\tau)$ for $\tau=(0,1)$, $(1,0)$, and $(1,1)$.
  (e--f) The perturbation with signature $\sigma=(-1,0)$, and clusters $C(S,\tau)$
  for $\tau=(0,0)$ and $\tau=(0,1)$, where $S$ is the union of the perturbed points (full dots),
  and grid points (empty circles).
  (g--h) The perturbation with signature $\sigma=(1,0)$ and clusters $C(S,\tau)$
   for $\tau=(1,0)$ and $\tau=(1,1)$ and the same $S$.}
\label{fig:perturbation}
\end{figure}

\begin{observation}\label{obs:1}
Consider a point set $S_i$ with signature $\sigma(i)\in \{-1,0,1\}^d$.
The clusters $C(S_i,\tau)$ cover $S_i$ if and only if for all $j=1,\ldots d$,
\begin{itemize}\itemsep -2pt
\item $\sigma_j(i)=0$, or
\item $\sigma_j(i)=-1$ and $\tau_j=0$ (e.g., Fig.~\ref{fig:perturbation}(e-f) for $j=1$), or
\item $\sigma_j(i)=1$ and $\tau_j=1$ (e.g., Fig.~\ref{fig:perturbation}(g-h ) for $j=1$).
\end{itemize}
\end{observation}

In the sequence of signatures $\sigma(1),\ldots ,\sigma(\lfloor{d/2}\rfloor)$,
we always change one zero coordinate to a nonzero coordinate.
In particular any nonzero coordinate $\sigma_j(i)\in \{-1,1\}$ remains unchanged,
that is, $\sigma_j(i)=\sigma_j(i+1)=\ldots =\sigma(\lfloor{d/2}\rfloor)$.
Consequently, for all $i=1,\ldots , \lfloor{d/2}\rfloor$,
if the clusters $C(S_1,\tau)$ cover $S_i$ for some $\tau \in \{0,1\}^d$, they also cover $S_{\leq i}$.
In particular, the clusters $C(S_1,\tau)$ that cover $S_{\lfloor d/2\rfloor}$
also cover the final point set $S_{\leq \lfloor{d/2}\rfloor}$.

\paragraph{Adversary strategy.}
At the end of round $i=1,\ldots , \lfloor d/2\rfloor-1$,
we compute $\sigma(i+1)$ from $\sigma(i)$ by changing a 0-coordinate to $-1$ or $+1$,
based on the clusters created by the $\varrho$-competitive online algorithm.
Note that every point in $S_i$, $i=1,2,\ldots,\lfloor d/2\rfloor$, has
$i-1$ perturbed coordinates and $d+1-i$ unperturbed coordinates.
For all points in $S_i$, all unperturbed coordinates are integers.
The algorithm covers $S_i$ with at most $\varrho \cdot \opt_i$ clusters.
Let $\pi_i: \RR^d\rightarrow \RR^{d+1-i}$ be the orthogonal projection
to the subspace spanned by the $d+1-i$ unperturbed coordinate axes.
Then every point in $S_1$ and its corresponding perturbations in $S_2,\ldots, S_i$
project to the same point with integer coordinates in $\RR^{d+1-i}$.
This implies that $\pi_i(S_i)=\pi_i(S_1)\subset \ZZ^{d+1-i}$,
and the projection $\pi_i(C)$ of a cluster $C$ contains at most
$2^{d+1-i}$ points in $\pi_i(S_i)$, that is, $|\pi_i(C\cap S_i)| \leq 2^{d+1-i}$.
A cluster $C$ created by the algorithm is called
\begin{itemize}\itemsep -2pt
\item \emph{small} if $|\pi_i(C\cap S_i)| \leq \frac{2^{d+1-i}}{2\varrho}=\frac{2^{d-i}}{\varrho}$,
  and
\item \emph{big} otherwise.
\end{itemize}
Note that we distinguish between small and big clusters in round $i$ with respect to the $d+1-i$
unperturbed coordinates;
in particular, a small cluster in round $i$ may become large in another round, or vice versa.
As we shall see, small clusters contain few points and can be ignored for now.
Big clusters, however, are constrained by the points they contain, and they cannot expand
in some of the coordinate directions. The adversary will present points in the next round that cannot
be added to most of the big clusters, rendering these clusters useless in subsequent rounds.

The $L_{\infty}$-diameter of a cluster and any of its projections to a subspace spanned by coordinate axes is most 1.
Consequently, a small cluster contains at most $(2^{d-i}/\varrho)\cdot 2^{i-1}=2^d/(2\varrho)$
points of $S_i$. Indeed, for each small cluster $C$, the projection $\pi_i(C)$ contains at most $2^{d-i}/\varrho$ points
in $\pi_i(S_i)$. Each point in $\pi_i(S_i)$ is the image of $K^{i-1}$ points of $S_i$;
since $S_i$ is a perturbation of the integer grid, any cluster contains at most $2^{i-1}$ of these preimages.
The total number of points in $S_i$ that lie in small clusters is at most (see~\eqref{eq:opt_i})
$$ (\varrho\cdot \opt_i) \ \frac{2^d}{2\varrho}
= \opt_i\cdot 2^{d-1}
= \left(\frac{1}{2}+o(1)\right) K^d. $$
Consequently, the remaining $\left(\frac{1}{2}-o(1)\right) K^d$ points in $S_i$ are covered
by big clusters. As any unit cluster contains at most $2^d$ points in $S_i$, 
the number of big clusters is at least
\begin{equation}\label{eq:big}
\left(\frac{1}{2}-o(1)\right) \frac{K^d}{2^d}.
\end{equation}
For a cluster $C$, let $s(C)$ denote the number of unperturbed coordinates
in which its extent is $1$. Then the projection $\pi_i(C)$
contains at most $2^{s(C)}$ integer points in $\ZZ^{d+1-i}$, hence $|\pi_i(C\cap S_i)|\leq 2^{s(C)}$.
Comparing the lower and upper bounds on the cardinality of $\pi_i(C\cap S_i)$, for a big cluster $C$, yields
\begin{eqnarray}
2^{d-i}/\varrho &<& 2^{s(C)} \nonumber\\
d-i-\log_2 \varrho &< & s(C). \nonumber
\end{eqnarray}

Consider the following experiment:
choose one of the zero coordinates of the signature $\sigma(i)$ uniformly
at random (\ie, all $d+1-i$ choices are equally likely), and change it to $-1$ or $+1$
with equal probability $1/2$.
Observe that if the $j$-th extent of a cluster $C$ is $1$, then it cannot be expanded in dimension $j$.
We say that a big cluster $C$ \emph{expires} if no point can (or will) be added to $C$
in the future. Recall that $i\leq \lfloor d/2\rfloor$ and we assume that $\varrho\leq d$.
Consequently, a big cluster $C$ expires with probability at least
\begin{equation} \label{eq:frac}
\frac{s(C)}{d+1-i}\cdot \frac{1}{2}
> \frac{d-i-\log_2 \varrho}{2(d+1-i)}
\geq \frac{d-\lfloor d/2\rfloor-\log_2 d}{2d}
= \Omega(1).
\end{equation}
Combining \eqref{eq:big} and~\eqref{eq:frac}, the expected number of clusters that expire
is
\[ \Omega(1)\cdot \left(\frac{1}{2}-o(1) \right) \frac{K^d}{2^d} = \Omega(\opt) .\]
It follows that there exists an unperturbed coordinate $j$, and a perturbation of it such that
$\Omega(\opt)$ big clusters expire at the end of round $i=1,\ldots , \lfloor d/2\rfloor -1$.
The adversary derandomizes the above experiment: it precomputes all possible perturbations
(changing a 0-coordinate of $\sigma(i)$ to $-1$ or $+1$),
and makes the choice that maximizes the number of big clusters that expire in that round.

In round $i= \lfloor d/2\rfloor$, all clusters that cover any point in $S_{\lfloor d/2\rfloor}$
expire, because no point will be added to any of these clusters.
Since $S_{\lfloor d/2\rfloor}$ is a perturbation of $S_1$, at least
$\opt_1=\Omega(\opt)$ clusters expire in the last round, as well.

If a cluster expires in round $i$, then it contains some points of $S_i$ but
does not contain any point of $S_j$ for $j>i$.
Consequently, each cluster expires in at most one round, and
the total number of expired clusters over all $\lfloor d/2\rfloor$ rounds is $\Omega(d \cdot \opt)$.
Since each of these clusters was created by the algorithm in one of the rounds, we have
$\varrho \cdot \opt = \Omega(d\cdot \opt)$, which implies $\varrho = \Omega(d)$, as claimed.

\paragraph{Randomized Algorithm.}
We modify the above argument to establish a lower bound of $\Omega(d)$ for
a randomized algorithm with an oblivious (randomized) adversary.
The adversary starts with the integer grid $S_1=[K]^d$, with signature $\sigma(1)=\mathbf{0}$
as before. At the end of round $i=1,\ldots ,\lfloor d/2\rfloor-1$, it chooses
an unperturbed coordinate of $\sigma(i)$ uniformly at random, and switches it to $-1$ or $+1$
with equal probability (independently of the clusters created by the algorithm)
to obtain $\sigma(i+1)$.
By~\eqref{eq:frac}, the expected number of big clusters that expire in round $i$,
$1 \leq i<\lfloor d/2\rfloor$,  is $\Omega(\opt_i)=\Omega(\opt)$; and all
$(1-o(1))\opt_{\lfloor d/2\rfloor} = \Omega(\opt)$ big clusters expire in round $\lfloor d/2\rfloor$.
Consequently, the expected number of clusters created by the algorithm is $\Omega(d\cdot \opt)$,
which implies $\varrho = \Omega(d)$, as required.
\end{proof}

\section{Lower bound for online  \textsc{Unit Covering} in $\RR^d$} \label{sec:unit-covering-r^d}

The following theorem extends a result from~\cite{CZ09} from $d=1$ to higher dimensions.

\begin{theorem} \label{thm:2^d}
The competitive ratio of every deterministic online algorithm for \textsc{Unit Covering}
in $\RR^d$ under the $L_{\infty}$ norm is at least $2^d$ for every $d \geq 1$.
\end{theorem}
Recall that \texttt{Algorithm Grid} attains a competitive ratio of $2^d$.
As such, \texttt{Algorithm Grid} is optimal in this setting, and
the lower bound in Theorem~\ref{thm:2^d} cannot be improved.
\begin{proof}
Consider a deterministic online algorithm $\alg$ for \textsc{Unit Covering}
in $\RR^d$. We present an input instance $\sigma$ for which the solution
$\alg(\sigma)$ is at least $2^d$ times $\opt(\sigma)$.
In particular, $\sigma$ consists of $2^d$ points in $\mathbb{R}^d$ that fit in a unit cube,
hence $\opt(\sigma)=1$, and we show that $\alg$ is required to place
a new unit cube for each point in $\sigma$.
Our proof works like a two player game between Alice and Bob. Here,
Alice is presenting points to Bob, one at a time. If a new cube is required,
Bob (who plays the role of the algorithm) decides where to place it.
Alice tries to force Bob to place as many new cubes as possible by presenting the points
in a smart way. Bob tries to place new cubes in a way such that they may cover other points
presented by Alice in the future, thereby reducing the need of placing new cubes quite often.

Throughout the game, Alice maintains a sequence of axis-aligned cubes $Q_1 \subset Q_2 \subset \ldots$,
each of side-length less than $1$, and Bob places axis-aligned cubes $U_1,U_2,\ldots$
to cover points presented by Alice.

Let $Q_0=U_0 = \emptyset$.
In step $i$, $i=1,\ldots,2^d$, Alice obtains $Q_i$ from $Q_{i-1}$,
where $Q_{i-1} \subset Q_i$.
More precisely, $Q_i$ is obtained by scaling up (slightly) $Q_{i-1}$ from a vertex.
This transformation defines a one-to-one correspondence between the vertices of $Q_i$ and $Q_{i+1}$
(as in Lemma~\ref{lem:exposed} that follows).
Alice then presents an arbitrary vertex of $Q_i$ that is not covered as the next point $p_i \in \sigma$,
and Bob covers it by placing the unit cube $U_i$.

For every $i\in \mathbb{N}$, let $\delta_i=2^{-2i}$ and $x_i=1-2\delta_i$
(in particular, $x_1=1/2$, $x_2=7/8$, and $x_3=31/32$).
For $i=1,\ldots , 2^d$, the side length of cube $Q_i$ is equal to $x_i$.
Note that $(x_i)_{i\in \mathbb{N}}$ is a strictly increasing sequence converging to $1$.

In step~1, Alice chooses $Q_1$ as an arbitrary cube of side-length $x_1$,
and the first point $p_1$ as an arbitrary vertex of $Q_1$.
Next, Bob places $U_1$ to cover $p_1$.
The remaining points $p_i$, for $i=2,\ldots,2^d$, in $\sigma$ are chosen adaptively,
depending on Bob's moves.

By the end of step~$i$, for $i=1,\ldots,2^d$, Alice has placed points $p_1,\ldots, p_i$,
and Bob has placed unit cubes $U_1,\ldots ,U_i$ (one for each of these points).
An illustration of the planar version of the game appears in Fig.~\ref{fig:f1}.

\begin{figure}[htpb]
\centering
\includegraphics[scale=0.6]{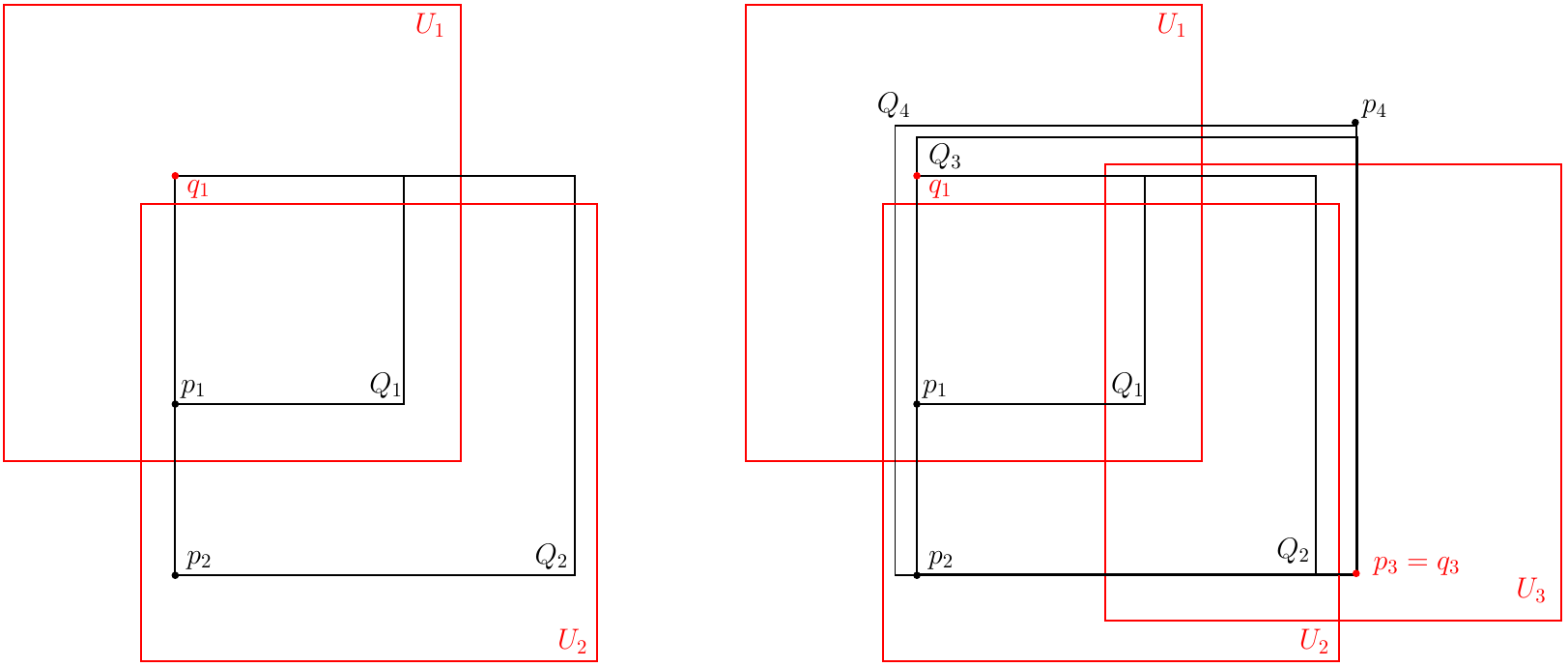}
\caption{A lower bound of $2^d$ on the competitive ratio.
  The figure illustrates the case $d=2$.
  Left: The first two points in $\sigma$ arrive. Right: the last two points in $\sigma$ arrive.
  The cubes placed by Bob ($U_1,U_2,U_3$)
  and the vertices that are deeply covered ($q_1$ and $q_3$)
  are colored in red.}
\label{fig:f1}
\end{figure}		

A vertex $v_i$ of $Q_i$ is said to be  \emph{covered} at time $t$
if $v_i$ is contained in the union of the cubes placed before time $t$.
Otherwise, $v_i$ is \emph{exposed}  (\ie, \emph{not covered}) at time $t$.
Note that in step~$1$, all $2^d$ vertices of $Q_1$ are exposed until Bob places $U_1$.

We maintain the following two invariants for $i=1,\ldots,2^d$:

\begin{enumerate}[(I)] \itemsep 1pt
\item \label{inv:1} the cube $Q_i$ contains the points $p_1,\ldots,p_i$;
\item \label{inv:2} the cube $Q_i$ has at least $2^d-i+1$ exposed vertices
       until $U_i$ is placed in step $i$
       (\ie, the union $\bigcup_{j<i} U_j$ contains at most $i-1$ vertices of $Q_i$).
\end{enumerate}
Invariant~\eqref{inv:2} ensures that Alice can present an exposed vertex of $Q_i$ in steps $i=1,\ldots,2^d$.
An exposed vertex $v_i$ of $Q_i$ is said to be \emph{deeply covered} by $U_i$ in step $i$
if $v_i$ is contained in $U_i$ and its distance from the boundary of $U_i$ is  larger than $\delta_i=(1-x_i)/2$;
\ie, $v_i\in U_i$ and $\dist(v_i, \partial U_i)>\delta_i$.
As we shall see, a deeply covered vertex helps producing exposed vertices in the next round.

\begin{lemma} \label{lem:deep}
  For $i \in \{1,\ldots,2^d-1\}$, at most one exposed vertex of $Q_i$ is deeply covered by
  $U_i$ in step $i$.
\end{lemma}
\begin{proof}
  Assume, to the contrary, that $u_i$ and $v_i$ are two exposed vertices of $Q_i$
  (in step $i$), that are deeply covered by $U_i$. Since $u_i$ and $v_i$ differ in at least
  one coordinate, the extent of $U_i$ in that coordinate is larger than
  $$ x_i + \frac{2(1-x_i)}{2} =1,$$
  which is a contradiction.
\end{proof}

If no exposed vertex of $Q_i$ is deeply covered by $U_i$, let
$Q_{i+1}$ be the unique axis-aligned cube of side length $x_{i+1}$
that contains $Q_i$ and has $p_i$ as a vertex (\ie, $Q_{i+1}$ is obtained by scaling up $Q_i$ from $p_i$).
Otherwise, let $q_i$ be the unique exposed vertex of $Q_i$
that is deeply covered by $U_i$ (possibly, $q_i=p_i$);
and let $Q_{i+1}$ be the unique axis-aligned cube of side length $x_{i+1}$
that contains $Q_i$ and has $q_i$ as a vertex (\ie, $Q_{i+1}$ is obtained by scaling up $Q_i$ from $q_i$).

\begin{lemma} \label{lem:exposed}
For $i \in \{1,\ldots,2^d-1\}$,
let $v_i$ be an exposed vertex of $Q_i$ in step $i$ that is
not deeply covered by $U_i$, and not the common vertex of $Q_i$ and $Q_{i+1}$.
Let $v_{i+1}$ be the vertex of $Q_{i+1}$ corresponding to $v_i$.
Then $v_{i+1}$ is not covered by $\bigcup_{j=1}^i U_j$.
\end{lemma}
\begin{proof}
  We claim that $v_{i+1}$ is not covered by $U_i$. First assume that $v_i\notin U_i$.
  We obtain $v_{i+1}$ from $v_i$ by scaling up $Q_i$ from a vertex ($p_i$ or $q_i$) in $U_i$.
  Since $U_i$ is convex, both $v_i$ and $v_{i+1}$ are in the exterior of $U_i$,
  and so $v_{i+1}$ is not covered by $U_i$.  Next assume that $v_i\in U_i$.
  Since $v_i$ is not deeply covered by $U_i$, its distance to the boundary of $U_i$ is at most $\delta_i$.
  By construction, there exist parallel faces of $Q_i$ and $Q_{i+1}$ that are incident to $v_i$ and $v_{i+1}$,
  respectively, but are not incident to the common vertex of $Q_i$ and $Q_{i+1}$. Two such faces are at distance
\[ x_{i+1} - x_i = 2(\delta_i - \delta_{i+1}) = 3 \cdot 2^{-(2i+1)} = \frac {3\delta_i}{2} > \delta_i \]
  from each other. This implies that $v_{i+1}$ is not covered by $U_i$.

  Since $v_i$ was not covered by the union of previous cubes $\bigcup_{j<i}U_j$ and
  all previous cubes intersect $Q_i$, it follows that the ray $\overrightarrow{v_i v_{i+1}}$
  does not intersect any previous cube $U_j$, $j<i$.
In particular, $v_{i+1}$ is not covered by $\bigcup_{j<i}U_j$.
As such, $v_{i+1} \not \in U_i$, hence $v_{i+1} \not \in \bigcup_{j=1}^i U_j$, as claimed.
\end{proof}

Invariant~\eqref{inv:1} follows inductively, by construction.
Invariant~\eqref{inv:2} follows inductively from
Lemmata~\ref{lem:deep} and~\ref{lem:exposed}.
By Invariant~\eqref{inv:2},
Alice can choose an exposed vertex of $Q_i$ in steps $i=1,\ldots,2^d$,
and Bob is required to place a new cube after each of these $2^d$ points,
hence $\alg(\sigma)=2^d$.
Invariant~\eqref{inv:1} yields $\opt(\sigma)=1$.
This completes the proof of Theorem~\ref{thm:2^d}.
\end{proof}

\section{Online \textsc{Unit Covering} in $\ZZ^d$} \label{sec:unit-covering-z^d}

\paragraph{Lower bound.}
When we consider unit covering over the integer lattice $\mathbb{Z}^d$, $d\geq 1$, the
adversary has fewer choices, and the $2^d$ lower bound of Theorem~\ref{thm:2^d} no longer
applies. Here we prove a lower bound that is linear in $d$.

\begin{theorem} \label{thm:d+1}
The competitive ratio of every deterministic online algorithm for \textsc{Unit Covering}
in $\ZZ^d$ under the $L_{\infty}$ norm is at least $d+1$ for every $d\geq 1$.
\end{theorem}
\begin{proof}
We construct an input sequence $p_1,\ldots,p_{d+1} \in \ZZ^d$ for which $\opt=1$ and $\alg=d+1$
using an adaptive adversary. We construct such a sequence inductively, so that
\begin{itemize} \itemsep -2pt
\item each new point $p_i$ requires a new cube, $Q_i \subset \RR^d$, and
\item all points presented can be covered by one integer unit cube incident to the origin.
\end{itemize}

Let $x_1,\ldots,x_d$ be the $d$ coordinate axes in $\RR^d$; and $x_{d+1}$ be the new axis in $\RR^{d+1}$.
The induction basis is $d=1$. We may assume for concreteness that $p_1=0$,
and suppose that the algorithm opens a unit interval $[x,x+1]$ to cover this point.
If $x=-1$, let $p_2=1$, else let $p_2=-1$. The algorithm now opens a new unit interval to cover $p_2$.
It is easily seen that $p_1,p_2 \in \ZZ$ and $\{p_1,p_2\}$ define a unit interval.

For the induction step, assume the existence of a sequence $\sigma=p_1,\ldots,p_{d+1} \in \ZZ^d$
that forces the algorithm to open a new unit cube, $Q_i \subset \RR^d$, to cover each new point $p_i$,
$i=1,\ldots,d+1$ (and so $\alg=d+1$), while $\opt=1$ with $\sigma$ being covered by a single cube
$U_d \subset \ZZ^d$.
Present a sequence of $d+2$ points to the algorithm in $\RR^{d+1}$;
the first $d+1$ points of this sequence are: $(p_1,0),\ldots,(p_{d+1},0)$.
The algorithm must use $d+1$ cubes,
say, $Q_1,\ldots,Q_{d+1} \subset \RR^{d+1}$ to cover these points. As such,
the $d+1$ unit cubes $\pi(Q_1),\ldots,\pi(Q_{d+1}) \subset \RR^d$, cover
$p_1,\ldots,p_{d+1} \in \ZZ^d$, where $\pi(Q_i)$ is the projection onto the first
$d$ coordinates of $Q_i$;
moreover, the unit cubes $\pi(Q_1),\ldots,\pi(Q_d)$ do not cover $p_{d+1}$.
Only $\pi(Q_d)$ contains $p_{d+1}$, but the cube $Q_d$ cannot contain
both $(p_{d+1},-1)$ and $(p_{d+1},1)$. Consequently, $(p_{d+1},-1)$ or $(p_{d+1},1)$
is not covered by $\bigcup_{i=1}^{d+1} Q_i$. The adversary presents such an exposed point,
which requires a new cube $Q_{d+2}$.
(Note that the points $p_1,\ldots,p_{d+2}$ form a lattice path, where $p_i$ and $p_{i+1}$
differ in the $(i+1)$-th coordinate.)
This completes the inductive step, and thereby the proof of the theorem.
\end{proof}

\paragraph{Upper bound.}
We substantially improve on the $2^d$ upper bound on the competitive ratio of
\textsc{Unit Covering} in $\RR^d$ (achieved by {\tt Algorithm Grid})
when the input points are in $\ZZ^d$
(or the $2^{d-1}+\frac12$ upper bound of {\tt Algorithm Greedy} in Section~\ref{sec:greedy}).

The online algorithm by Buchbinder and Naor~\cite{BN09b} for \textsc{Set Cover},
for the unit covering problem over $\ZZ^d$, yields an algorithm with $O(d \log{(n/\opt)})$
competitive ratio under the assumption that a set of $n$ possible integer points
is given in advance. Recently, Gupta and Nagarajan~\cite{GN14} gave an online randomized
algorithm for a broad family of combinatorial optimization problems that can be expressed
as sparse integer programs. For unit covering over the integers in $[n]^d$, their results
yield a competitive ratio of $O(d^2)$, where $[n]=\{1,2,\ldots,n\}$. The competitive ratio
does not depend on $n$, but the algorithm must know $n$ in advance.
We remark that \emph{if} the algorithm is allowed to maintain several candidate solutions
and return a best candidate at termination (which is customary in data stream models~\cite{LLR18}),
then this approach combined with a standard randomized shifting method~\cite{HM85}
would yield a competitive ratio of $O(d^2)$. However, in the online model we consider here,
the algorithm can maintain only one solution, and this approach is no longer viable.

We now remove the dependence on $n$ so as to get a truly online algorithm for
\textsc{Unit Covering} over $\ZZ^d$. Consider the following randomized algorithm.

\begin{quote}
\texttt{Algorithm Iterative Reweighing.}
Let $P\subset \ZZ^d$ be the set of points presented to the
algorithm and $\C$ the set of cubes chosen by the algorithm;
initially $P=\C=\emptyset$.
The algorithm chooses cubes for two different reasons,
and it keeps them in sets $\C_1$ and $\C_2$, where $\C=\C_1\cup \C_2$.
It also maintains a third set of cubes, $\B$, for bookkeeping purposes;
initially $\B=\emptyset$.
In addition, the algorithm maintains a weight function on all integer unit cubes.
Initially $w(Q)=2^{-(d+1)}$ for all integer unit cubes
(this is the default value for all cubes that are disjoint from $P$).

We describe one iteration of the algorithm.
Let $p\in \ZZ^d$ be a new point; put $P\leftarrow P\cup \{p\}$.
Let $\Q(p)$ be the set of $2^d$ integer unit cubes that contain $p$.
\begin{enumerate}\itemsep -2pt
\item If $p\in \bigcup \C$, then do nothing.
\item Else if $p\in \bigcup \B$, then let $Q\in \B\cap \Q(p)$ be an arbitrary cube
  and put $\C_1 \leftarrow \C_1\cup \{Q\}$.
\item\label{setp:3} Else if $\sum_{Q\in \Q(p)} w(Q) \geq 1$, then let $Q$ be an arbitrary cube
  in $\Q(p)$ and put $\C_2\leftarrow \C_2\cup \{Q\}$.
\item\label{step:4} Else, the weights give a probability distribution on $\Q(p)$.
  Successively choose cubes from $\Q(p)$ at random with this distribution
  in $2d$ independent trials and add them to $\B$.
  Let $Q \in \B\cap \Q(p)$ be an arbitrary cube and put $\C_1 \leftarrow \C_1\cup \{Q\}$.
  Double the weight of every cube in $\Q(p)$.
\end{enumerate}
\end{quote}

\begin{theorem} \label{thm:covering-integer}
  The competitive ratio of \texttt{Algorithm Iterative Reweighing} for \textsc{Unit Covering}
  in $\ZZ^d$ under the $L_{\infty}$ norm is $O(d^2)$ for every $d\in \NN$.
\end{theorem}
\begin{proof}
Suppose that a set $P$ of $n$ points is presented to the algorithm sequentially,
and the algorithm created unit cubes in $\C=\C_1\cup \C_2$. Note that $\C_1 \subseteq \B$.
We show that $\EE[|\B|]= O(d^2\cdot \opt)$ and $\EE[|\C_2|]=O(\opt)$.
This immediately implies
that $\EE[|\C|] \leq \EE[|\C_1|] + \EE[|\C_2|] \leq \EE[|\B|] + \EE[|\C_2|] =O(d^2\cdot \opt)$.

First consider $\EE[|\B|]$. New cubes are added to $\B$ in step~\ref{step:4}. In this case,
the algorithm places at most $2d$ cubes into $\B$, and doubles the weight of all $2^d$ cubes
in $\Q(p)$ that contain $p$. Let $\C_{\opt}$ be an offline optimum set of unit cubes.
Each point $p\in P$ lies in some cube $Q_p \in \C_{\opt}$.
The weight of $Q_p$ is initially $2^{-(d+1)}$, and it never exceeds 2; indeed,
since $Q_p \in \Q(p)$, its weight before the last doubling must have been at most $1$
in step~\ref{step:4} of the algorithm; thus its weight is doubled in at most $d+2$ iterations.
Consequently, the algorithm invokes step~\ref{step:4} in at most $(d+2) \opt$ iterations.
In each such iteration, it adds at most $2d$ cubes to $\B$. Overall, we have
$|\B|\leq (d+2)\cdot 2d \cdot \opt = O(d^2\cdot \opt)$, as required.

Next consider $\EE[|\C_2|]$. A new cube is added to $\C_2$ in step~\ref{setp:3}.
In this case, none of the cubes in $\Q(p)$ is in $\B$ and $\sum_{Q\in \Q(p)} w(Q) \geq 1$
when point $p$ is presented, and the algorithm increments $|\C_2|$ by one.
At the beginning of the algorithm, we have
$\sum_{Q\in \Q(p)} w(Q) = \sum_{Q\in \Q(p)} 2^{-(d+1)}= 2^d\cdot 2^{-(d+1)}=1/2$.
Assume that the weights of the cubes in $\Q(p)$ were increased in $t$ iterations,
starting from the beginning of the algorithm, and the sum of weights of the cubes
in $\Q(p)$ increases by $\delta_1,\ldots , \delta_t>0$
(the weights of several cubes may have been doubled in an iteration).
Since $\sum_{Q\in \Q(p)} w(Q)=1/2+\sum_{i=1}^t \delta_i$,
then $\sum_{Q\in \Q(p)} w(Q) \geq 1$ implies $\sum_{i=1}^t \delta_i \geq 1/2$.
For every $i=1,\ldots, t$, the sum of weights of some cubes in $\Q(p)$,
say, $\Q_i \subset \Q(p)$, increased by
$\delta_i$ in step~\ref{step:4} of a previous iteration. Since the weights doubled,
the sum of the weights of these cubes was $\delta_i$ at the beginning of that iteration,
and the algorithm added one of them into $\B$ with probability at least $\delta_i$
in one random draw, which was repeated $2d$ times independently. Consequently, the probability
that the algorithm did \emph{not} add any cube from $\Q_i$ to $\B$ in that iteration is at most
$(1-\delta_i)^{2d}$. The probability that none of the cubes in $\Q(p)$ has been
added to $\B$ before point $p$ arrives is (by independence) at most
$$\prod_{i=1}^t (1-\delta_i)^{2d}
\leq e^{-2d\sum_{i=1}^t \delta_i}
\leq e^{-d}.$$
The total number of points $p$ for which step~\ref{setp:3} applies is at most $|P|$.
Since each unit cube contains at most $2^d$ points, we have $|P|\leq 2^d\cdot \opt$.
Therefore $\EE[|\C_2|] \leq |P| e^{-d} \leq (2/e)^d \, \opt \leq \opt$, as claimed.
\end{proof}

The above algorithm applies to \textsc{Unit Clustering} of integer points in $\ZZ^d$
with the same competitive ratio:
\begin{corollary} \label{cor:clustering-integer}
  The competitive ratio of \texttt{Algorithm Iterative Reweighing} for \textsc{Unit Clustering}
  in $\ZZ^d$ under the $L_{\infty}$ norm is $O(d^2)$ for every $d\in \NN$.
\end{corollary}

\section{Lower bound for \texttt{Algorithm Greedy} for \textsc{Unit Clustering}} \label{sec:greedy}

Chan and Zarrabi-Zadeh~\cite{CZ09} showed that the greedy algorithm
for \textsc{Unit Clustering} on the line ($d=1$) has competitive ratio of $2$
(this includes both an upper bound on the ratio and a tight example).
Here we show that the competitive ratio of the greedy algorithm is unbounded for $d\geq 2$.
We first recall the algorithm:

\begin{quote}
\texttt{Algorithm Greedy.}
For each new point $p$, if $p$ fits in some existing cluster, put $p$ in such a cluster
(break ties arbitrarily); otherwise open a new cluster for $p$.
\end{quote}

\begin{theorem} \label{thm:greedy}
  The competitive ratio of \texttt{Algorithm Greedy} for \textsc{Unit Clustering}
  in $\RR^d$ under the $L_{\infty}$ norm is unbounded for every $d\geq 2$.
\end{theorem}
\begin{proof}
It suffices to consider $d=2$; the construction extends to arbitrary dimensions $d\geq 2$.
The adversary presents $2n$ points in pairs
$\{ (1+i/n,i/n), (i/n,1+i/n)\}$ for $i=0,1,\ldots ,n-1$.
Each pair of points spans a unit square that does not contain any subsequent point.
Consequently, the greedy algorithm will create $n$ clusters, one for each point pair.
However, $\opt=2$ since the clusters $C_1=\{ (1+i/n,i/n):i=0,1,\ldots ,n-1\}$ and
$C_2=\{(i/n,1+i/n):i=0,1,\ldots ,n-1\}$ are contained in the unit squares
$[1,2]\times [0,1]$ and $[0,1]\times [1,2]$, respectively.
\end{proof}

\noindent
When we restrict \texttt{Algorithm Greedy} to integer points, its competitive ratio is
exponential in $d$.

\begin{theorem} \label{thm:greedy-integer}
  The competitive ratio of \texttt{Algorithm Greedy} for \textsc{Unit Clustering} in $\ZZ^d$
  under the $L_{\infty}$ norm is at least $2^{d-1}$ and at most $2^{d-1}+\frac12$ for every $d \geq 1$.
\end{theorem}
\begin{proof}
We first prove the lower bound.
Consider an integer input sequence implementing a barycentric subdivision of the space,
as illustrated in Fig.~\ref{fig:greedy}.
Let $K$ be a sufficiently large positive multiple of $4$ (that depends on $d$).
We present a point set $S$, where $|S|=(2+o(1))(K/2)^d $ points to the algorithm,
and show that it creates $(1+o(1)) \, 2^{d-1} \opt$ clusters.
\begin{figure}[htbp]
\centering
\includegraphics[scale=0.67]{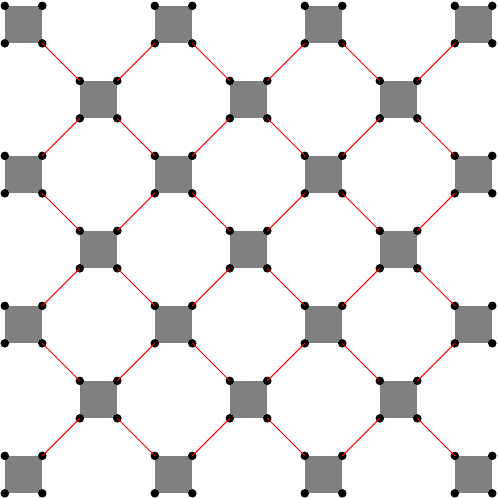}
\caption{A planar instance for the greedy algorithm with $K=12$; the edges in $E$ are
  drawn in red.}
\label{fig:greedy}
\end{figure}

Let $S=B \cup D$, where
\begin{align*}
A &= \{(x_1,\ldots,x_d)~|~ x_i \equiv 0 \pmod 4, \ 0 \leq x_i \leq K, \ i=1,\ldots,d\}, \\
B &= A + \{0,1\}^d,\\
C &= \{(x_1,\ldots,x_d)~|~ x_i \equiv 2 \pmod 4, \ 0 \leq x_i \leq K, \ i=1,\ldots,d\}, \\
D &= C + \{0,1\}^d, \\
E &= \{\{u,v\}: u \in B, v \in D, ||u-v||_\infty \leq 1\}.
\end{align*}
Note that each element of $C$ is the barycenter (center of mass) of $2^d$ elements
of $A$, namely the vertices of a cell of $(4\ZZ)^d$ containing the element.
Here $E$ is a set of pairs of lattice points (edges) that can be put in one-to-one
correspondence with the points in $D$. As such, we have
\begin{align*}
|A| &= \left(\frac{K}{4}+1\right)^d, \ |B|= 2^d |A| =(1+o(1)) \, \frac{K^d}{2^d}, \\
|C| &= \left(\frac{K}{4}\right)^d, \ |D|= 2^d |C| =(1+o(1)) \, \frac{K^d}{2^d}, \\
|E| &= |D| = (1+o(1)) \, \frac{K^d}{2^d}, \\
\opt &=|A \cup C|=|A|+|C|= (2+o(1)) \left(\frac{K}{4}\right)^d.
\end{align*}

It follows that $|E|=(1+o(1)) \, 2^{d-1} \opt$.
The input sequence presents the points in pairs, namely those in $E$.
The greedy algorithm makes one new non-extendable cluster for each such ``diagonal'' pair
(each cluster is a unit cube), so its competitive ratio is at least $2^{d-1}$ for every $d\geq 2$.

An upper bound of $2^d$ follows from the fact that each cluster in $\opt$ contains at most
$2^d$ integer points; we further reduce this bound. Let $\Gamma_1,\ldots,\Gamma_k$ be
the clusters of an optimal partition ($k=\opt$). Assume that the algorithm produces
$m$ clusters of size at least $2$ and $s$ singleton clusters. Since each cluster of $\opt$
contains at most one singleton cluster created by the algorithm, we have
\begin{align*}
  \alg &=m +s \leq \frac{(k-s) 2^d + s(2^d -1)}{2} + s =
  \frac{k \, 2^d -s}{2} + s \\
  &= k \,2^{d-1} + \frac{s}{2} \leq k \,2^{d-1} + \frac{k}{2} = k\left(2^{d-1} + \frac12\right),
\end{align*}
as required.
\end{proof}

\section{Conclusions} \label{sec:conclusion}

Our results in Theorems~\ref{thm:unit-clustering} and~\ref{thm:2^d}
show that the competitive ratio for both online \textsc{Unit  Clustering}
and \textsc{Unit Covering} grows with the dimension of the space.
From a broader perspective, the main question is how the curse of dimensionality plays
out for online algorithms in this area. In principle, the growth rates in the dimension
may be different for \textsc{Unit  Clustering} and \textsc{Unit Covering}.

On the one hand, the tight bound obtained for \textsc{Unit Covering} in $\RR^d$
shows that the growth rate of the competitive ratio for this problem must be exponential.
On the other hand, currently no online algorithm is known for \textsc{Unit Clustering} in $\RR^d$
under the $L_\infty$ norm with a competitive ratio $o(2^d)$. The current best upper bound for the competitive ratio
under this norm is $2^d  \cdot \frac56$ for sufficiently large dimensions $d \geq 2$, which is only
marginally better than the trivial $2^d$ ratio. This evidence suggests that the growth rate of the
competitive ratio for this problem may be exponential, as well. The additional degree of flexibility
(or ``ambiguity'') in \textsc{Unit Clustering} may be a reason for the difficulty in obtaining
a better lower bound if this is how things will finally turn out.

\smallskip
Several directions remain for future study. We summarize a few specific questions of interest.

\begin{question}
  Is there an upper bound of $o(2^d)$ on the competitive ratio for
  \textsc{Unit Clustering} in $\RR^d$ under the $L_\infty$ norm?
\end{question}

\begin{question}
  Is there a lower bound on the competitive ratio for
  \textsc{Unit Clustering} that is exponential in $d$?
  Is there a superlinear lower bound?
\end{question}

For online \textsc{Unit Covering} in $\RR^d$ under the $L_\infty$ norm,
the competitive ratio of the deterministic {\tt Algorithm Grid} is $2^d$, which is the best possible.
One remaining issue is in regard to randomized algorithms and oblivious adversaries.

\begin{question}
  Is there an upper bound of $o(2^d)$ on the competitive ratio of randomized algorithms for
  \textsc{Unit Covering} in $\RR^d$ under the $L_\infty$ norm?
\end{question}

\begin{question}
  Is there a superlinear lower bound on the competitive ratio of randomized algorithms
  (against oblivious adversaries) for \textsc{Unit Covering} in $\RR^d$ under the $L_\infty$ norm?
\end{question}

The reader is also referred to~\cite{Du18} for a discussion of related problems.


\begin{thebibliography}{99}

\bibitem{AAA+09}
Noga Alon, Baruch Awerbuch, Yossi Azar, Niv Buchbinder, and Joseph Naor,
The online set cover problem,
\emph{SIAM J. Comput.} \textbf{39(2)} (2009), 361--370.

\bibitem{ABC+16}
Yossi Azar, Niv Buchbinder, T.-H. Hubert Chan, Shahar Chen, Ilan Reuven Cohen,
Anupam Gupta, Zhiyi Huang, Ning Kang, Viswanath Nagarajan, Joseph Naor, and Debmalya Panigrahi,
Online algorithms for covering and packing problems with convex objectives,
in \emph{Proc. 57th IEEE Symp.\ on Foundations of Computer Science (FOCS)}, IEEE, 2016,
pp.~148--157.

\bibitem{ABFP13}
Yossi Azar, Umang Bhaskar, Lisa Fleischer, and Debmalya Panigrahi,
Online mixed packing and covering,
in \emph{Proc. 24th ACM-SIAM Symposium on Discrete Algorithms (SODA)}, SIAM, 2013, pp.~85--100.

\bibitem{ACR17}
Yossi Azar, Ilan Reuven Cohen, and Alan Roytman,
Online lower bounds via duality,
in \emph{Proc. 28th ACM-SIAM Symposium on Discrete Algorithms (SODA)}, SIAM, 2017, pp.~1038--1050.

\bibitem{BBKTW94}
Shai Ben-David, Allan Borodin, Richard M. Karp, G\'abor Tardos, and Avi Wigderson,
On the power of randomization in on-line algorithms,
\emph{Algorithmica} \textbf{11(1)} (1994), 2--14.

\bibitem{BLMS17}
Ahmad Biniaz, Peter Liu, Anil Maheshwari, and Michiel Smid,
Approximation algorithms for the unit disk cover problem in 2D and 3D,
\emph{Comput. Geom.} \textbf{60} (2017), 8--18.

\bibitem{BY98}
Allan Borodin and Ran El-Yaniv,
\emph{Online Computation and Competitive Analysis},
Cambridge University Press, Cambridge, 1998.

\bibitem{BMP05}
Peter Brass, William Moser, and J\'anos Pach,
\emph{Research Problems in Discrete Geometry}, Springer, New York, 2005.

\bibitem{BN09b}
Niv Buchbinder and Joseph Naor,
Online primal-dual algorithms for covering and packing,
\emph{Math. Open. Res.} \textbf{34(2)} (2009), 270--286.

\bibitem{CZ09}
Timothy M. Chan and Hamid Zarrabi-Zadeh,
A randomized algorithm for online unit clustering,
\emph{Theory Comput. Syst.} \textbf{45(3)} (2009), 486--496.

\bibitem{CCFM04}
Moses Charikar, Chandra Chekuri, Tom\'as Feder, and Rajeev Motwani,
Incremental clustering and dynamic information retrieval,
\emph{SIAM J. Comput.} \textbf{33(6)} (2004), 1417--1440.

\bibitem{Ch08}
Marek Chrobak,
SIGACT news online algorithms column 13,
\emph{SIGACT News Bulletin} \textbf{39(3)} (2008), 96--121.

\bibitem{CEIL13}
J\'anos Csirik, Leah Epstein, Csan\'ad Imreh, and Asaf Levin,
Online clustering with variable sized clusters,
\emph{Algorithmica} \textbf{65(2)} (2013), 251--274.

\bibitem{DI13}
Gabriella Div\'eki and Csan\'ad Imreh,
An online 2-dimensional clustering problem with variable sized clusters,
\emph{Optimization and Engineering} \textbf{14(4)} (2013), 575--593.

\bibitem{DI14}
Gabriella Div\'eki and Csan\'ad Imreh,
Grid based online algorithms for clustering problems,
in \emph{Proc. 15th IEEE Int. Sympos. Comput. Intel. Infor. (CINTI)}, IEEE, 2014, pp.~159.

\bibitem{Du18}
Adrian Dumitrescu,
Computational geometry column 68,
\emph{SIGACT News} \textbf{49(4)}, (2018), 46--54.

\bibitem{DGT18}
Adrian Dumitrescu, Anirban Ghosh, and Csaba D. T\'oth,
Online unit covering in Euclidean space,
\emph{Theoretical Computer Science} \textbf{809} (2020), 218--230.

\bibitem{EL13}
Martin R. Ehmsen and Kim S. Larsen,
Better bounds on online unit clustering,
\emph{Theoret. Comput. Sci.} \textbf{500} (2013), 1--24.

\bibitem{ELS08}
Leah Epstein, Asaf Levin, and Rob van Stee,
Online unit clustering: Variations on a theme,
\emph{Theoret. Comput. Sci.} \textbf{407(1-3)} (2008), 85--96.

\bibitem{ES10}
Leah Epstein and Rob van Stee,
On the online unit clustering problem,
\emph{ACM Trans. Algorithms} \textbf{7(1)} (2010), 1--18.

\bibitem{FG88}
Tom\'as Feder and Daniel H. Greene,
Optimal algorithms for approximate clustering,
in \emph{Proc. 20th ACM Symposium on Theory of Computing (STOC)},
1988, pp. 434--444.

\bibitem{FPT81}
Robert J. Fowler, Mike Paterson, and Steven L. Tanimoto,
Optimal packing and covering in the plane are NP-complete,
\emph{Inform. Process. Lett.} \textbf{12(3)} (1981), 133--137.

\bibitem{Go85}
Teofilo F. Gonzalez,
Clustering to minimize the maximum intercluster distance,
\emph{Theoret. Comput. Sci.} \textbf{38} (1985), 293--306.

\bibitem{GN14}
Anupam Gupta and Viswanath Nagarajan,
Approximating sparse covering integer programs online,
\emph{Math. Oper. Res.} \textbf{39(4)} (2014), 998--1011.

\bibitem{HM85}
Dorit S. Hochbaum and Wolfgang Maass,
Approximation schemes for covering and packing problems in image processing and VLSI,
\emph{J. ACM} \textbf{32(1)} (1985), 130--136.

\bibitem{KK15}
Jun Kawahara and Koji M. Kobayashi,
An improved lower bound for one-dimensional online unit clustering,
\emph{Theoret. Comput. Sci.} \textbf{600} (2015), 171--173.

\bibitem{LLR18}
Christopher Liaw, Paul Liu, and Robert Reiss,
Approximation schemes for covering and packing in the streaming model,
in \emph{Proc. 30th Canadian Conference on Computational Geometry (CCCG)},
Winnipeg, MB, 2018, pp.~172--179.

\bibitem{MS84}
Nimrod Megiddo and Kenneth J. Supowit,
On the complexity of some common geometric location problems,
\emph{SIAM J. Comput.} \textbf{13(1)} (1984), 182--196.

\bibitem {Va01} Vijay Vazirani,
\emph{Approximation Algorithms},
Springer Verlag, New York, 2001.

\bibitem {WS11} David P. Williamson and David B. Shmoys,
\emph{The Design of Approximation Algorithms},
Cambridge University Press, 2011.

\bibitem{ZC09}
Hamid Zarrabi-Zadeh and Timothy M. Chan,
An improved algorithm for online unit clustering,
\emph{Algorithmica} \textbf{54(4)} (2009), 490--500.

\end{thebibliography}
\end{document}